%% file: main.tex
\documentclass[
reprint,
 amsmath,amssymb,
 prb,
]{revtex4}

\usepackage{graphicx}% Include figure files
\usepackage{dcolumn}% Align table columns on decimal point
\usepackage{bm}% bold math

\usepackage[dvipsnames]{xcolor}
\usepackage{subcaption}
\usepackage{hyperref}% add hypertext capabilities
\graphicspath{{figs/}}

\usepackage{afterpage}
\usepackage{amsmath}
\usepackage{amsthm}
\usepackage{amssymb}
\usepackage{csquotes}
\usepackage{enumitem}
\usepackage{graphicx}
\usepackage{lipsum}
\usepackage{booktabs}
\usepackage{dsfont}
\usepackage{mathtools}

\usepackage{subcaption}

\usepackage{bm}
\usepackage[normalem]{ulem}

\hypersetup{
  linkcolor  = black,
  citecolor  = black,
  urlcolor   = black,
  colorlinks = true,
}

\newcommand{\I}{\mathrm{i}}

\theoremstyle{definition}
\newtheorem{definition}{Definition}[section]

\theoremstyle{plain}
\newtheorem{proposition}{Proposition}[section]

\newtheorem{theorem}{\bf Theorem}[section]

\theoremstyle{remark}
\newtheorem*{remark}{Remark}

\begin{document}
\title{Edge modes in modulated metamaterials based on the three-gap theorem}
\input{authors.tex}
\date{\today}

\begin{abstract}
We present a new paradigm for generating complex structured materials based on the three-gap theorem that unifies and generalises several key concepts in the study of localised edge states. Our model has both the discretised coupling strengths of the Su--Schrieffer--Heeger model and a modulation parameter that can be used to characterise the spectral flow of edge modes and produce images reminiscent of the Hofstadter butterfly. By defining a \textit{localisation factor} associated to each eigenmode, we are able to establish conditions for the existence of localised edge states in finite systems. This allows us to compare their eigenfrequencies with the spectra of the corresponding infinitely periodic problem and characterise the rich pattern of localised edge modes appearing and disappearing (in the sense of becoming delocalised) as the parameters of our three-gap algorithm are varied.
\end{abstract}
\maketitle

\section{Introduction}
Understanding and controlling wave localisation in complex media has been one of the main objectives of wave physics in recent decades. Advances in micro- and nano-scale manufacturing have allowed the scientific community to create more intricate and complex structures, achieving outstanding wave control properties even in deeply subwavelength regimes \cite{smith2004metamaterials, craster2012acoustic, soukoulis2011past}. Nonetheless, wave localisation remains difficult to predict in general, especially when a material has highly intricate small-scale structure.

One of the most fundamental examples of wave localisation is the occurrence of \emph{edge states}. These are eigenstates which are localised at the edges of a finite-sized piece of a material. Edge states are observed in many different settings but are particularly prevalent at the edges of materials which support spectral gaps. These gaps correspond to energies of waves which are not able to travel through the material and whose energy is typically reflected back towards the source. As a result, if an eigenstate falls at an energy within a gap, we can expect it to be localised. The challenge, however, is to predict when these eigenstates exist. There are two main approaches that have been considered in the literature: topological indices and spectral flow.

The generation of edge states due to topological indices is a well-studied phenomenon in the setting of periodic systems. The canonical example of this is the Su--Schrieffer--Heeger (SSH) model \cite{su1980soliton} which features alternating coupling strengths between discrete elements (inspired by electron dynamics in the presence of the alternating coupling strengths of polyacetylene \cite{su1980soliton}). The properties of these localised states are intimately related to the underling topology of the periodic material, which can be characterised in terms of a topological index known as the \emph{Zak phase} \cite{zak1989berry}. This can be related to the surface impedance, which can in turn be used to prove the existence of edge states \cite{xiao2014surface, thiang2023bulk, coutant2024surface, alexopoulos2023topologically}. The key feature of the SSH model (two alternating coupling strengths) can be extended naturally to models with three different couplings, which are often known as ``SSH3" \cite{anastasiadis2022bulk, martinez2019edge, verma2024emergent, zhang2021topological} and whose topological properties can likewise be used to characterise the existence of localised modes.

Another popular tool for predicting the existence of edge states is to calculate the spectral flow associated to edge modes in a gap. That is, keeping track of the net number of eigenvalues that cross a spectral gap as a parameter of the system is varied. This approach has gained popularity recently as it gives a direct intuitive approach to understand not only the existence of edge states, but also insight on how to tune their properties \cite{chaplain2023tunable}. It has been applied to periodic systems \cite{miniaci2023spectral, ammari2022robust, drouot2020defect, drouot2021bulk} and also generalised to non-periodic structures, especially those based on modulated periodic systems \cite{ni2019observation, xia2020topological} (which, it should be noted, can also be handled using topological approaches derived from K-theory \cite{apigo2018topological}). In the case of these modulated systems, the parameter that is varied when considering the spectral flow is the phase of the modulation relative to the periodicity of the lattice \cite{ni2019observation, xia2020topological}. This modulation is sometimes known as ``Harper modulation" after Harper's equation \cite{harper1955single} and the behaviour of the spectrum as a function of the modulated phase parameter often leads to fractal plots that are reminiscent of the Hofstadter butterfly \cite{hofstadter1976energy}.

In this work, we present a new paradigm for generating structured materials that unites the key geometric principles of both the SSH model and Harper modulation. We exploit the three-gap theorem, which says that when an arbitrary number of points are placed around a circle at equal intervals, there are only ever at most three different distances between adjacent points (even when the placement algorithm has wrapped around the circle multiple times) \cite{Allouche_Shallit_2003}. This presents an approach to generating geometries that only ever have at most three distinct distances (like the three distinct coupling strengths of SSH3) and have a phase parameter that can be naturally modulated to study the spectral flow of edge states (in this case, the fixed angle used in the placement algorithm).

The three-gap geometries proposed here could be considered in many different physical settings. In this work, we consider a canonical example of a mass-spring lattice where the spring constants are taken to be inversely proportional to the distance between adjacent points. We will study the Floquet-Bloch spectra of periodic materials with unit cells generated by the three-gap algorithm and then examine the existence of edge states when finite-sized lattices are considered with Dirichlet boundary conditions (\emph{i.e.} the masses at the ends are fixed). Our analysis will be based on studying a \emph{localisation factor} associated to each eigenmode. This tool quantifies the decay rate and allows us to make precise statements about the exponential localisation of eigenmodes in finite systems \cite{ammari2024exponentially}. This allows us to make initial steps towards unifying the ideas of SSH models and Harper modulated systems, which have previously been characterised by topological indices and notions of spectral flow, respectively.

This article is organised as follows: in Section \ref{sec:ProblemSetting} we will introduce the physical setting and the notation used for the three-gap algorithm on which the structured materials considered in this work are based. We will also state the three different eigenvalue problems that we will consider (which correspond to infinitely periodic systems and finite systems with one or more unit cells). Sections~\ref{gen:finite} and \ref{subsec:RelationToTheBulkSpectra} will prove some general statements about the existence of localised modes in finite systems and their relation to the spectra of the associated infinitely periodic systems. These results are not dependent on the specific geometrical and physical properties of our three-gap system. In Section~\ref{sec:DeterminationOfFloquetBlochAngles}, these general statements will be applied to the specific case of the three-gap geometries, focusing on the conditions that vanish the edge states from the spectra.

\section{Problem setting}
\label{sec:ProblemSetting}

Let us imagine a circle of unitary perimeter and pick a fixed angle $\theta \in [0,1)$. Starting from a given point (that we will call $0$ or the origin), we form a set of points on the circle by rotating $m$ times the angle $\theta$ for $m=0,1,\dots,N$. This is depicted in Figure~\ref{fig_scheme} for $N=4$. After generating the set of points on the circle, we map this onto a one-dimensional system by uncurling the circle into a straight line and rescaling so that the total length is $N$ (meaning the average density of points is constant, irrespective of $N$). This gives the set of points
\begin{equation} \label{eq:points}
   C_N^\theta=\{\theta_j = N\mathrm{frac}(j \theta):j=0,\dots,N-1\},
\end{equation}
where \(\mathrm{frac}\) denotes the fractional part of a real number, defined as $\mathrm{frac}(x)=x-\lfloor x\rfloor$. For ease of notation, this set is labelled in increasing order so that
\[C_N^\theta=\{x_j\}_{j=0}^{N-1},\]
and \(x_j\leq x_{j+1}\) for all \(j\). Notice that $x_0=0$. We also let \(x_N=N\).

\begin{figure}[!h]
\centering\includegraphics[width=0.8\linewidth]{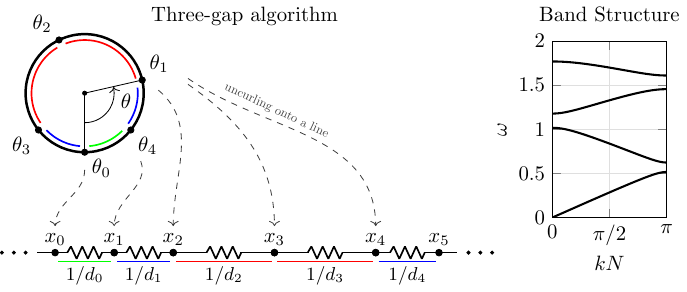}
\caption{Schematic of the three-gap algorithm for unit cell formation. Points on the circle are generated by repeatedly rotating a given angle $\theta$ and then uncurled onto a straight line. These points become the masses of a mass-spring system whose spring constants are inversely proportional to the distance between neighbouring points. On the right hand side, the spectral band structure corresponding to the infinitely periodic system is shown.}
\label{fig_scheme}
\end{figure}

We will convert the set of points \eqref{eq:points} generated by the three-gap algorithm into a mass-spring system by placing identical masses at each of the points and connecting them to their immediate neighbours with springs whose stiffnesses are inversely proportional to the distances between the two masses. That is, since the point \(x_j\) denotes the equilibrium position for the \(j\)-th mass, its neighbouring distance is \(d_j=x_{j+1}-x_j\) for \(j=0,\dots,N-1\). Then, the \(j\)-th spring (which connects the \(j\)-th and \(j+1\)-th masses) has spring constant \(1/d_j\). This is sketched in Figure~\ref{fig_scheme}.

Assuming all the masses have equal unit mass, by Newton's second law the motion of the $j$-th mass is described by the governing equation
\begin{align}
    \frac{d^2u_j}{dt^2} &= \frac{1}{d_j}(u_{j+1}-u_j)-\frac{1}{d_{j-1}}(u_j-u_{j-1})
    = \frac{1}{d_j}u_{j+1}-\left(\frac{1}{d_j}+\frac{1}{d_{j-1}}\right)u_j+\frac{1}{d_{j-1}}u_{j-1},
    \label{eq:DiffEquation}
\end{align}
where \(u_j(t)\) denotes the displacement from \(x_j\) of the \(j\)-th mass at time $t$. Depending on the boundary conditions applied on the displacement of masses $u_0$ and $u_N$ we will consider two different kinds of spectra: a continuous Floquet-Bloch spectrum when the system is periodic and a discrete spectrum when the system has finite size with fixed boundary conditions.

\subsection{Floquet-Bloch spectrum}

We, first, consider the spectrum of the infinitely periodic system formed by repeating the unit cell given by the $N$ masses from $C_N^\theta$. Since we are assuming the masses are identical, the infinite periodic mass-spring system is defined by the infinite periodic sequence of distances \(\{d_j\}_{j=-\infty}^\infty\) generated by the finite sequence \(\{d_j\}_{j=0}^{N-1}\), where \(d_i=d_j\) if \(i\equiv j \pmod N\). Then, the equilibrium equation gives
\[\frac{d^2u_j}{dt^2}=\frac{1}{d_j}u_{j+1}-\left(\frac{1}{d_j}+\frac{1}{d_{j-1}}\right)u_j+\frac{1}{d_{j-1}}u_{j-1}, \, j\in \mathbb{Z}.\]
Assuming we are in a time-harmonic regime, such that \(u_j(t)=v_j e^{-\I\omega t}\), the above expression becomes
\begin{equation}
\label{infinite_sys_diff}
    v_{j+1}=d_j\left(\frac{1}{d_j}+\frac{1}{d_{j-1}}-\omega^2\right)v_j-\frac{d_j}{d_{j-1}}v_{j-1}, \, j\in \mathbb{Z}.
\end{equation}

The theory of \textit{Floquet-Bloch} for periodic operators \cite{kuchment2016overview} tells us that propagating solutions to this problem (\emph{i.e.} those which don't grow or decay) can be expressed as
\begin{equation}
    v_{j+N}=v_j e^{\I Nk}, \, j\in \mathbb{Z},    
\end{equation}
for some \(k \in \mathbb{R}\). 
% 
% Since the infinite system in (\ref{infinite_sys_diff}) is linear with periodic coefficients, if there exits some \(j' \in \mathbb{Z}\), \(\beta \in \mathbb{C}\) such that
% \[
% \begin{pmatrix}
%     v_{j'+N+1}\\v_{j'+N}
% \end{pmatrix}=\beta
% \begin{pmatrix}
%     v_{j'+1}\\v_{j'}
% \end{pmatrix}
% \] then
% \[
% \begin{pmatrix}
%     v_{j+N+1}\\v_{j+N}
% \end{pmatrix}
% =\beta
% \begin{pmatrix}
%     v_{j+1}\\v_j
% \end{pmatrix}
% s\] for all \(j \in \mathbb{Z}\).
Hence, there exists a solution to \eqref{infinite_sys_diff} with \textit{Floquet exponent} \(k\) if and only if
\begin{equation}
\widetilde{K}(k)\boldsymbol{v}=\omega^2\boldsymbol{v}
\end{equation}
for some nontrivial \(\boldsymbol{v}=(v_0,v_1,\dots, v_{N-2},v_{N-1})^T\), where
\[
\widetilde{K}(k)=
\begin{pmatrix}
\frac{1}{d_0}+\frac{1}{d_{N-1}} & -\frac{1}{d_0} & 0 & \cdots & 0 & -\frac{1}{d_{N-1}} e^{-\I Nk}\\
-\frac{1}{d_0} & \frac{1}{d_0}+\frac{1}{d_1} & -\frac{1}{d_1} & \cdots & 0 & 0\\
0 & -\frac{1}{d_1} & \frac{1}{d_1}+\frac{1}{d_2} & \cdots & 0 & 0\\
\vdots & \vdots & \vdots & \ddots & \vdots & \vdots\\
0 & 0 & 0 & \cdots & \frac{1}{d_{N-3}}+\frac{1}{d_{N-2}} & -\frac{1}{d_{N-2}}\\
-\frac{1}{d_{N-1}} e^{\I Nk} & 0 & 0 & \cdots & -\frac{1}{d_{N-2}} & \frac{1}{d_{N-2}}+\frac{1}{d_{N-1}}
\end{pmatrix}.
\]
The dispersion curves of the infinite system are the solutions $(k,\omega)$ to the dispersion relation
\begin{equation}
\label{dispersion relation}
    \det(\widetilde{K}(k)-\omega^2I_N)=0.
\end{equation}
Some exemplar dispersion curves were shown in Figure~\ref{fig_scheme}. As the unit cell gets increasingly complex and contains an increasingly large number of points, the pattern of spectral bands and gaps will become increasingly intricate. However, some features can be predicted immediately, such as the fact there will always be a high-frequency band gap (\emph{i.e.} there will be some $\omega^*$ such that all $\omega>\omega^*$ are in a band gap) \cite{hori1964structureI,hori1964structureII,dunckley2024hierarchical}.

\subsection{Finite-sized array}

In the case of a finite-sized array, we suppose that the masses at either end are fixed such that \(u_0\equiv 0\) and \(u_N\equiv 0\) (\emph{i.e.} ``Dirichlet" boundary conditions). Therefore, the system of equations in (\ref{eq:DiffEquation}) is reduced to the $N-1$ unknown displacements of the masses with indices $j=1,..., N-1$ and it can be written in matrix form as 
\begin{equation}
    K\boldsymbol{v}=\omega^2\boldsymbol{v},
    \label{eq:GeneralEigenvalueProblem}
\end{equation}
where \(\boldsymbol{v}=(v_1,v_2,\dots, v_{N-1})^T\) are the harmonic displacement amplitudes $u_j(t)=v_j e^{-\I\omega t}$ and
% 
% \begin{equation}
%     \frac{d^2\boldsymbol{u}}{dt^2}=-K\boldsymbol{u},
%     \label{eq:DiffEquationTimeDependent}
% \end{equation}
% where \(\boldsymbol{u}=(u_1,u_2,\dots, u_{N-1})^T\),
\begin{equation}
\label{single unit cell}
    K=
\begin{pmatrix}
\frac{1}{d_0}+\frac{1}{d_1} & -\frac{1}{d_1} & 0 & \cdots & 0 & 0 \\
-\frac{1}{d_1} & \frac{1}{d_1}+\frac{1}{d_2} & -\frac{1}{d_2} & \cdots & 0 & 0 \\
0 & -\frac{1}{d_2} & \frac{1}{d_2}+\frac{1}{d_3} & \cdots & 0 & 0 \\
\vdots & \vdots & \vdots & \ddots & \vdots & \vdots \\
0 & 0 & 0 & \cdots & \frac{1}{d_{N-3}}+\frac{1}{d_{N-2}} & -\frac{1}{d_{N-2}} \\
0 & 0 & 0 & \cdots & -\frac{1}{d_{N-2}} & \frac{1}{d_{N-2}}+\frac{1}{d_{N-1}}
\end{pmatrix}.
\end{equation}
Hence, \(\omega^2\) is an eigenvalue of \(K\) with eigenvector \(\boldsymbol{v}\). Since we have a linear ordinary differential equation with real coefficients, the general solution is
\[\boldsymbol{u}(t)=\sum_{\omega^2\in \sigma(K)} \boldsymbol{v}_\omega (c_\omega e^{\I\omega t}+\overline{c_\omega} e^{-\I\omega t}),\]
where \(\boldsymbol{v}_\omega\) is an \(\omega^2\)-eigenvector, and \(c_\omega\) is a complex constant.

We will also consider finite-sized arrays formed by repeating a $C_N^\theta$ unit cell several times. Consider some fixed values for $N\in\mathbb{N}$ and $\theta\in(0,1)$ and consider a system formed by repeating $n$ times the unit cell $C_N^\theta$. The eigenstates of this system are solutions to the eigenvalue problem
\begin{equation} \label{multiple unit cell}
K_n \boldsymbol{v}=\omega^2 \boldsymbol{v},
\end{equation}

where $K_n$ is the $(nN-1)\times(nN-1)$ square matrix given by
\begin{equation*}
K_n=
\begin{pmatrix}
    K & -\frac{1}{d_{N-1}}\\
    -\frac{1}{d_{N-1}} & \frac{1}{d_{N-1}}+\frac{1}{d_0} & -\frac{1}{d_0}\\
    & -\frac{1}{d_0} & K & -\frac{1}{d_{N-1}}\\
    & & -\frac{1}{d_{N-1}} & \frac{1}{d_{N-1}}+\frac{1}{d_0} & -\frac{1}{d_0}\\
    & & & -\frac{1}{d_0} & K & \ddots\\
    & & & & \ddots & \ddots & -\frac{1}{d_{N-1}}\\
    & & & & & -\frac{1}{d_{N-1}} & \frac{1}{d_{N-1}}+\frac{1}{d_0} & -\frac{1}{d_0}\\
    & & & & & & -\frac{1}{d_0} & K\\
\end{pmatrix}.
\end{equation*}
Here, \(K=K_1\) is the $(N-1)\times(N-1)$ matrix for single unit cell system given in (\ref{single unit cell}).

In summary, there are three different systems whose spectra we will consider in this work. First, a finite-sized system with the unit cell based on $C_N^\theta$ repeated $n$ times. The spectrum of this system are the eigenvalues of $K_n$ from \eqref{multiple unit cell}. We will study the edge states in this system for the case of fixed (Dirichlet) boundary conditions. Second, we will consider the same system with just one $C_N^\theta$ unit cell. In this case, the spectrum is the eigenvalues of $K$, as defined in \eqref{single unit cell}. Note that, trivially, $K=K_1$. However, we will see below that the restriction to a single unit cell is not just a simplification but that the eigenvalues of the single unit cell have fundamental implications for edge states in larger systems. Finally, to better understand and characterise the edge states we will compare the spectra of $K_n$ and $K$ to the spectrum of the infinitely periodic system, whose continuous spectral bands are given by the eigenvalues of $\widetilde{K}(k)$ as $k$ varies.

%---------------------------------------------------------%
%---------------------------------------------------------%
\section{Existence of localised modes}
\label{sec:ExistenceOfLocalisedModes}

In this section, we will develop theory to characterise the existence of solutions to the finite-sized array problem whose eigenstates are localised at one edge of the finite cluster. Our focus will be on the finite-sized system with $n$ unit cells (characterised by the eigenvalues of the matrix $K_n$). However, we will need to study the eigenvalues of $K$ (corresponding to a single unit cell) and $\widetilde{K}(k)$ (corresponding to a periodic system) in order to properly understand the edge states. Statements in this section and in subsection \ref{subsec:RelationToTheBulkSpectra} are general to any kind of mass-spring system, and do not depend on the procedure followed to create the geometry of the system.

\subsection{General properties of finite-sized systems} \label{gen:finite}

\begin{proposition}
\label{simple eval}
    The matrix \(K_n\) has \(nN-1\) distinct eigenvalues.
\end{proposition}

\begin{proof}
    Recall that tridiagonal matrices with nonzero super- (or sub-) diagonal have nullity either \(0\) or \(1\). Then, the eigenvalues for \(K_n\) have geometric multiplicity \(1\). Since \(K_n\) is real symmetric, \(K_n\) is diagonalisable. Thus, eigenvalues for \(K_n\) are real with algebraic multiplicity \(1\). Hence we conclude that \(K_n\) has \(nN-1\) distinct eigenvalues.
\end{proof}

Suppose that \(\omega^2\) is an eigenvalue for the finite-sized system with \(n\) unit cells, \emph{i.e.} an eigenvalue of \(K_n\). We call the unique (up to sign) normalised eigenvector 
\begin{equation}
    \boldsymbol{v}^n_\omega=(v^n_1, \dots, v^n_{nN-1})^T,
\end{equation}
for \(\omega^2\) the \textit{mode} for frequency \(\omega^2\). If $\boldsymbol{v}^n_\omega$ is normalised such that $\sup_{1\leq j\leq nN-1} |v^n_j|=1$, then we say the mode is \textit{localised} at an edge of the system if \(v^n_1 \to 0\) or \(v^n_{nN-1} \to 0\) as \(n \to \infty\). Note that for the notion of localisation to be well-defined, we need to ensure that \(\omega^2\) is an eigenvalue of \(K_n\) for all \(n \in \mathbb{N}\). With this in mind, we make the following definition:

\begin{definition}
\label{fixed_eval_def}
    We say an eigenvalue \(\omega^2\) for a finite \((N,\theta)\) system is \textit{universal} if \(\omega^2\) is an eigenvalue of \(K_n\) for all \(n \in \mathbb{N}\).
\end{definition}

\begin{proposition}
\label{fixed K-modes}
    If \(\omega^2\) is an eigenvalue for the single unit cell system, then \(\omega^2\) is universal.
\end{proposition}

\begin{proof}
    Fix \(n \in \mathbb{N}\). It suffices to show that \(c_K(x) \mid c_{K_n}(x)\) where \(c(x)\) denotes the characteristic polynomial. From the Laplace expansion of the determinant, we have that \[c_{K_n}(x)=\left(x-\left(\frac{1}{d_0}+\frac{1}{d_{N-1}}\right)\right)c_K(x)c_{K_{n-1}}(x)-\frac{1}{d_{N-1}^2}p(x)c_{K_{n-1}}(x)-\frac{1}{d_0^2}c_K(x)q_{n-1}(x)\]
    for some \(p(x), q_{n-1}(x) \in \mathbb{R}[x]\). Hence, by induction, \(c_K(x) \mid c_{K_n}(x)\).
\end{proof}

For any eigenvalue \(\omega^2\) of the single unit cell system in (\ref{single unit cell}), we denote \[\boldsymbol{v}_\omega=(v^\omega_1, \dots, v^\omega_{N-1})^T\] as the unique \(\omega^2\)-mode. Let $\alpha(\omega)$ be the \textit{localisation factor} defined for the $\omega^2$-mode as  
\begin{equation}
    \alpha(\omega)=-\frac{d_0}{d_{N-1}}\frac{v^\omega_{N-1}}{v^\omega_1},
\end{equation}
then we have the following theorem, which shows how localisation at either edge can be characterised by the value of $\alpha(\omega)$.

\begin{theorem} \label{thm:alpha1}
    Fix \((N,\theta) \in \mathbb{N}_{\geq 2} \times (0,1)\). For \(\omega^2 \in \sigma(K)\) and \(n \in \mathbb{N}\), the \(\omega^2\)-mode for \(K_n\) (the finite \((N,\theta)\) system with \(n\) unit cells) is localised if and only if \(\lvert \alpha(\omega) \rvert \neq 1\).
\end{theorem}

\begin{proof}
    Fix \(n \in \mathbb{N}\). By Proposition~\ref{fixed K-modes}, we have \(\omega^2 \in \sigma(K_n)\). As written in (\ref{multiple unit cell}), \(K_n\) is almost block diagonal if there are no ``cross"-like entries between the \(K\) blocks. Hence, it is natural to assume an ansatz for an eigenvector of the form
    \[\boldsymbol{v}=(v_{1,1}, \dots, v_{1,N-1},0,v_{2,1}, \dots, v_{2,N-1},0,\dots, v_{n,1}, \dots, v_{n,N-1})^T,\]
    where \(\boldsymbol{v}_j=(v_{j,1}, \dots, v_{j,N-1})^T\) is an \(\omega^2\)-eigenvector of \(K\) for \(j \in \{1, \dots, n\}\). By Proposition \ref{simple eval}, we can write
    \[\boldsymbol{v}_j=k_j \boldsymbol{v}_\omega,\]
    for \(k_j \in \mathbb{R}\), where \(j \in \{1, \dots, n\}\). For \(\boldsymbol{v}\) to be an eigenvector, we need
    \[-\frac{1}{d_{N-1}}v_{j,N-1}-\frac{1}{d_0}v_{j+1,1}=0,\]
    for all \(j \in \{1, \dots, n-1\}\).
    That is
    \[\frac{k_{j+1}}{k_j}=\frac{v_{j+1,1}}{v_{j,1}}=-\frac{d_0}{d_{N-1}}\frac{v_{j,N-1}}{v_{j,1}}=-\frac{d_0}{d_{N-1}}\frac{v^\omega_{N-1}}{v^\omega_1}=\alpha(\omega).\]
    Then,
    \[k_j=\alpha(\omega)^{j-1}k_1,\]
    for \(j \in \{1, \dots, n\}\). Since
    \[\lVert \boldsymbol{v} \rVert=\max_{0\leq j \leq n-1}\lvert\alpha(\omega)\rvert^j \, \lVert k_1\boldsymbol{v}_\omega \rVert=\max\{1,\lvert\alpha(\omega)\rvert^{n-1}\} \, \lvert k_1 \rvert,\]
    for \(\boldsymbol{v}\) to be the \(\omega^2\)-mode, we need
    \[\lvert k_1 \rvert=\frac{1}{\max\{1,\lvert\alpha(\omega)\rvert^{n-1}\}}.\]
    By the uniqueness (up to sign) of \(\omega^2\)-mode,
    \[\boldsymbol{v}^n_\omega=(k_1\boldsymbol{v}_\omega \, , \, 0 \, , \, \alpha(\omega)k_1\boldsymbol{v}_\omega \, , \, 0 \, , \,  \dots \, , \, \alpha(\omega)^{n-1}k_1\boldsymbol{v}_\omega)^T.\]
    Hence the \(\omega^2\)-mode is localised if and only if \(\lvert \alpha(\omega) \rvert \neq 1\). Moreover, if \(\lvert \alpha(\omega) \rvert<1\), then \(v^n_{nN-1} \to 0\) as \(n \to \infty\); if \(\lvert \alpha(\omega) \rvert>1\), then \(v^n_1 \to 0\) as \(n \to \infty\).
\end{proof}

\begin{remark}
    From the proof of Theorem~\ref{thm:alpha1} we can see that $\alpha(\omega)$ not only tells us whether an eigenmode is localised, but also indicates where it is localised. The eigenmode is localised at the $j=0$ edge of the finite-sized array if $|\alpha(\omega)|<1$ and it is localised at the $j=nN-1$ end if $|\alpha(\omega)|>1$.
\end{remark}

%---------------------------------------------------------%
%---------------------------------------------------------%
\subsection{Relation to Floquet-Bloch spectra}
\label{subsec:RelationToTheBulkSpectra}

To better understand the localisation of eigenstates, it is informative to compare the spectra of the finite-sized array with the Floquet-Bloch spectra of the infinitely periodic material. Consider the dispersion relation from (\ref{dispersion relation}):
\[\det(\omega^2I_N-\widetilde{K}(k))=0.\]
Using the fact that $\det(\widetilde{K}(k))=1$ for all $k$, this simplifies to
\begin{equation} \label{eq:dispersion_cosine}
\cos{Nk}=\frac{1}{2}(-1)^N \frac{c_{K_2}(\omega^2)}{c_K(\omega^2)}\prod_{j=0}^{N-1}d_j.
\end{equation}
Note that if \(c_K(\omega^2)=0\), then the ratio \(c_{K_2}(\omega^2)/c_K(\omega^2)\) is interpreted as \(\lim_{t \to \omega^2} c_{K_2}(t)/c_K(t)\). Clearly,  \eqref{eq:dispersion_cosine} has a real-valued solution for $\omega$ if and only if the right-hand side is less than one in absolute value. Hence, \(\omega^2 \in \mathbb{R}\) is in a \emph{band gap} (and there exists no Floquet solution) if and only if 
\begin{equation}
\frac{1}{2} \left\lvert\frac{c_{K_2}(\omega^2)}{c_K(\omega^2)}\right\rvert \prod_{j=0}^{N-1}d_j>1.
\end{equation}
This observation informs the following proposition, which says that if an eigenvalue of the finite-sized system with $n$ unit cells falls in a band gap of the periodic system, then that eigenvalue must be an eigenvalue of the finite-sized system with just one unit cell. For simplicity, we prove it for the special case where $n$ is a power of 2, although we conjecture that it holds in general.

\begin{proposition}
    Fix \(n=2^q\) for some \(q \in \mathbb{N}\).     
    If \(\omega^2 \in \sigma(K_n)\) and $\omega^2\notin\cup_k\sigma(\widetilde{K}(k))$, then \(\omega^2 \in \sigma(K)\).
    \label{prop:Proposition33}
\end{proposition}

\begin{proof}
    Define \(K_n^-\) and \(K_{n-}\) as submatrices of \(K_n\) by deleting the first row and column and the last row and column, respectively. We have
    \[c_{K_{2n}}(x)=c_{K_n}(x)\left(\left(x-\left(\frac{1}{d_0}+\frac{1}{d_{N-1}}\right)\right)c_{K_n}(x)-\frac{1}{d_{N-1}^2}c_{K_{n-}}(x)-\frac{1}{d_0^2}c_{K_n^-}(x)\right).\]
    We also have the generalised dispersion relation
    \[\cos{nNk}=\frac{1}{2}(-1)^{nN}\prod_{j=0}^{N-1}d_j^n\left(\left(x-\left(\frac{1}{d_0}+\frac{1}{d_{N-1}}\right)\right)c_{K_n}(x)-\frac{1}{d_{N-1}^2}c_{K_{n-}}(x)-\frac{1}{d_0^2}c_{K_n^-}(x)\right).\]
    Since \(\omega^2\) is in a gap,
    \[\frac{1}{2}\prod_{j=0}^{N-1}d_j^n\left\lvert\left(\omega^2-\left(\frac{1}{d_0}+\frac{1}{d_{N-1}}\right)\right)c_{K_n}(\omega^2)-\frac{1}{d_{N-1}^2}c_{K_{n-}}(\omega^2)-\frac{1}{d_0^2}c_{K_n^-}(\omega^2)\right\rvert>1,\]
    then
    \[\left(\omega^2-\left(\frac{1}{d_0}+\frac{1}{d_{N-1}}\right)\right)c_{K_n}(\omega^2)-\frac{1}{d_{N-1}^2}c_{K_{n-}}(\omega^2)-\frac{1}{d_0^2}c_{K_n^-}(\omega^2) \neq 0,\]
    hence \(c_{K_{2n}}(\omega^2)=0\) implies \(c_{K_n}(\omega^2)=0\). Thus, the statement follows by induction.
\end{proof}

% \begin{corollary_p}
%     If \(\omega^2\) is a universal eigenvalue that falls in a band gap of the periodic system, then \(\omega^2 \in \sigma(K)\).
% \end{corollary_p}

% \begin{proof}
%     It follows by Definition~\ref{fixed_eval_def}.
% \end{proof}

The above statement says that, given a finite-sized system with $n$ repetitions of the unit cell $C_N^\theta$, if an eigenvalue is in a gap of the periodic system, then it must be an eigenvalue of the finite-sized system with just one unit cell. As a result, in our search for localised states in systems with multiple unit cells, it will be sufficient to consider the system with just one unit cell. The natural next step is to investigate under which conditions the eigenvalues from the single unit cell system lie in a band gap of the periodic structure. Before proving the result, we recall the infinite system (\ref{infinite_sys_diff}) which can be written in transfer matrix form as
\begin{equation}
    \begin{pmatrix}
        v_{j+1}\\v_j
    \end{pmatrix}
    =M_j(\omega)
    \begin{pmatrix}
        v_j\\v_{j-1}
    \end{pmatrix},
\end{equation}
where \[M_j(\omega)=
\begin{pmatrix}
    d_j(\frac{1}{d_j}+\frac{1}{d_{j-1}}-\omega^2) & -\frac{d_j}{d_{j-1}}\\
    1 & 0
\end{pmatrix}.
\]
Let \[T(\omega)=\prod_{j=1}^N M_j(\omega),\]
then the characteristic polynomial of the $2\times2$ matrix $T$ is 
\[c_{T(\omega)}(x)=x^2-\text{tr}(T(\omega))x+1.\]
Note that if \((v_j)_{j \in \mathbb{Z}}\) is an \(\omega^2\)-Floquet solution to (\ref{infinite_sys_diff}) with Floquet exponent \(k\), then
\[T(\omega)\begin{pmatrix}
    v_1\\v_0
\end{pmatrix}=\begin{pmatrix}
    v_{N+1}\\v_{N}
\end{pmatrix}=e^{iNk}\begin{pmatrix}
    v_1\\v_0
\end{pmatrix}.\]
That is, \(e^{\pm iNk}\) are eigenvalues of \(T(\omega)\).

\begin{theorem} \label{Theorem:alphaGap}
    Fix \(\omega^2 \in \sigma(K)\), then \(\omega^2\) is in a band gap of the periodic system if and only if \(\lvert \alpha(\omega) \rvert \neq 1\).
\end{theorem}
\begin{proof}
    Note that \[(v_j)_{j \in \mathbb{Z}} \vcentcolon=\left(\alpha(\omega)^{[\frac{j}{N}]} v^\omega_{j-N[\frac{j}{N}]}\right)_{j \in \mathbb{Z}}=(\dots, \alpha(\omega)^{-1}v^\omega_{N-1},0,v^\omega_1, \dots, v^\omega_{N-1},0,\alpha(\omega)v^\omega_1, \dots)\] is an \(\omega^2\)-solution for the infinite system (\ref{infinite_sys_diff}). Then
    \[T(\omega) \begin{pmatrix}
        v^\omega_1\\0
    \end{pmatrix}=\alpha(\omega)\begin{pmatrix}
        v^\omega_1\\0
    \end{pmatrix}.\]
    Hence, \(T(\omega)\) has eigenvalues $\alpha(\omega)$ and $ \alpha(\omega)^{-1}$. We prove the contrapositive as follows. Suppose \(\lvert \alpha(\omega) \rvert=1\), then \((v_j)_{j \in \mathbb{Z}}\) is an \(\omega^2\)-Floquet solution (with \(k=0\) if \(\alpha(\omega)=1\) and \(k=\frac{\pi}{N}\) if \(\alpha(\omega)=-1\)). Conversely, suppose there exits an \(\omega^2\)-Floquet solution, then eigenvalues of \(T(\omega)\) have modulus \(1\), hence \(\lvert \alpha(\omega) \rvert=1\).
\end{proof}

Figure~\ref{fig_2} shows an example of how the spectra of the finite and infinitely periodic systems are related.  Panel (a) is a phase space diagram that shows how the eigenfrequencies $\omega^2$ vary when the phase parameter $\theta$ is varied for a modulated crystal with fixed $N=4$. Black lines represent the eigensolutions of the Floquet-Bloch spectra (recall that each $\theta$ corresponds to a different geometrical configuration, and thus, a different eigenvalue problem). Red lines represent eigensolutions for a finite problem of a crystal with $n = 7$ unit cells. It can be noticed that some of the eigenvalues from the $K_7$ problem lie in the same frequency range of the propagating bands of the spectrum of the infinite structure, while others lie in the bandgaps of the infinite structure. Those eigenvalues lying in the bandgap of the infinite structure are the eigensolutions of the single unit cell system $K$, as proved by Proposition \ref{prop:Proposition33}. Following Theorem~\ref{Theorem:alphaGap}, we can examine the localisation factor $\alpha(\omega)$ corresponding to one of these eigensolutions of the finite problem that lies in a gap of the infinite structure.  

\begin{figure}[!h]
\centering\includegraphics[width=0.85\linewidth]{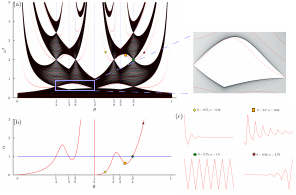}
\caption{Panel (a) is a phase space diagram showing eigenfrequencies $\omega^2$ vs. modulation parameter $\theta$ for the spectrum of the infinitely periodic system (black lines) and for the finite problem with $n=7$ unit cells (red lines). Blue dashed lines represent geometrical conditions for which the infinitely periodic and the finite spectrum share at least one of their eigenvalues. Panel (b) shows the evolution of the localisation factor $\alpha$ as a function of the modulation parameter for one of the eigenstates of the finite-sized problem that lies in the bandgap of the infinite system (the blue line indicates the critical value $|\alpha|=1$). Finally, panel (c) depicts four eigenstates corresponding to different geometrical structures that have been marked with colour markers in panels (a) and (b).}
\label{fig_2}
\end{figure}

Figure~\ref{fig_2}(b) shows how the localisation factor $\alpha$ corresponding to one of the eigenmodes of the finite system varies as a function of the modulation parameter $\theta$. The chosen eigenvalue is the one that has the coloured markers in panel \ref{fig_2}(a). The localisation factor $\alpha(\omega)$ crosses the critical line $\alpha(\omega) = 1$ at several different values of $\theta$. Those points coincide with those geometrical configurations for which the given eigensolution belongs both to the spectrum of the periodic system ($\sigma(\tilde{K})$) and the finite system ($\sigma(K)$). These points will be analysed in the following subsection. One may notice that the crossings $\alpha(\omega) = 1$ are a discrete set of points, and not continuous intervals. By taking a look at some parts of the phase space diagram in Figure~\ref{fig_2}(a) one could think that the overlapping between red and black lines is produced for a given interval. However, if we enlarge that part of the phase space diagram, as in the inset in the top right hand-side of Figure~\ref{fig_2}(a), it is clear that the overlap only occurs at a single degenerate point. 

Finally, Figure~\ref{fig_2}(c) shows the four eigenstates corresponding to the four eigenvalues shown with colour markers in panels (a) and (b). Those eigenstates having $\alpha(\omega)$ values that are very different from 1 are highly confined, as is the case for the first and the last eigenvectors. Whether the value of $\alpha(\omega)$ is much larger than or much smaller than 1 determines at which edge of the system the mode is localised. In the case of a localisation factor closer to one, the eigenstate still decreases from one edge to the other, but the envelope of the function decays more slowly than in the previous cases. When the localisation factor is equal to $1$, the eigenstate has constant amplitude across the structure, so that no localisation exists.

\begin{figure}[!h]
\centering
\begin{subfigure}{0.45\textwidth}
    \includegraphics[width=\textwidth,trim=0 0 0 2.5cm, clip]{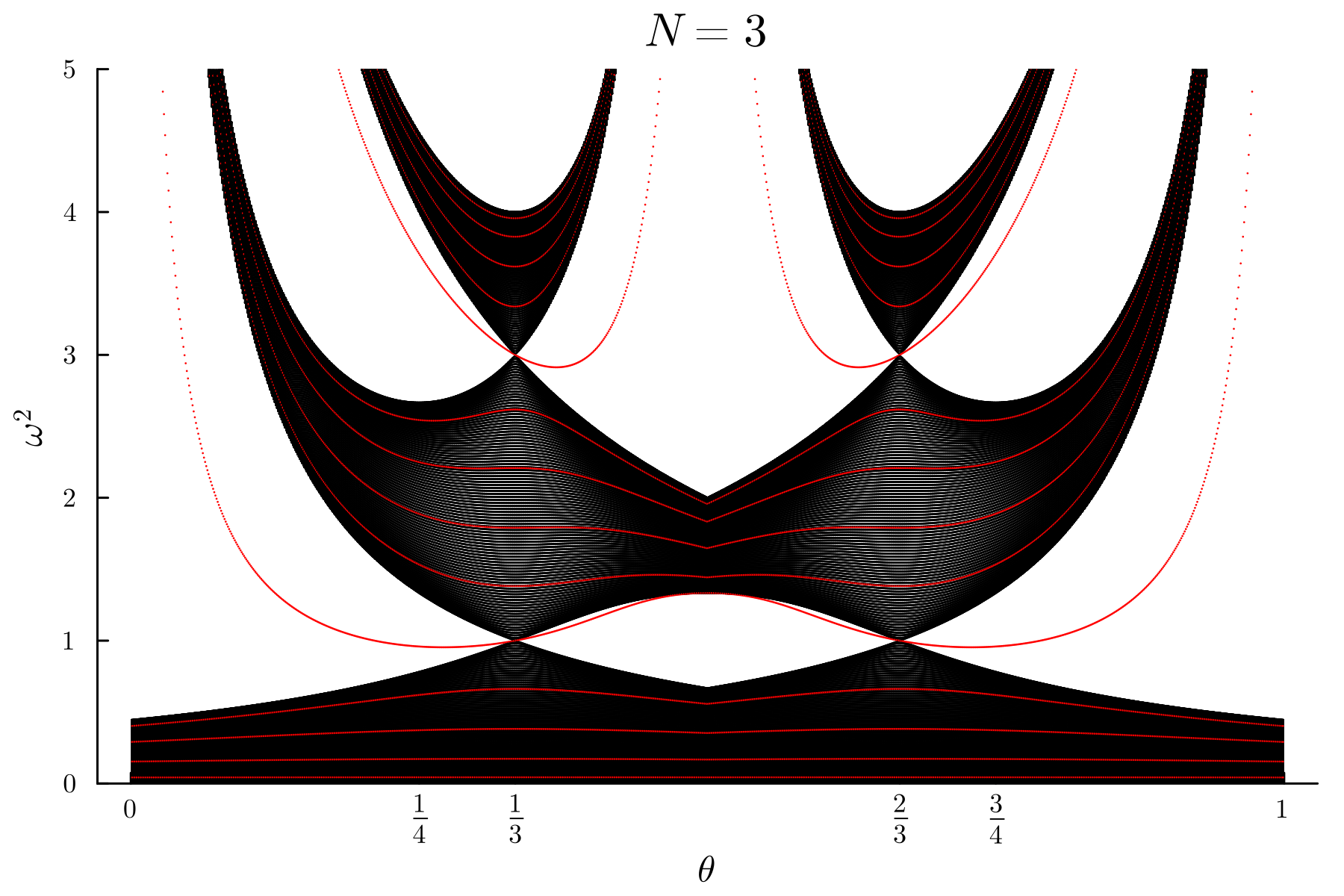}
    \caption{$N=3$}
    \label{fig:figure_3_panel_a}
\end{subfigure}
\begin{subfigure}{0.45\textwidth}
    \includegraphics[width=\textwidth,trim=0 0 0 3cm, clip]{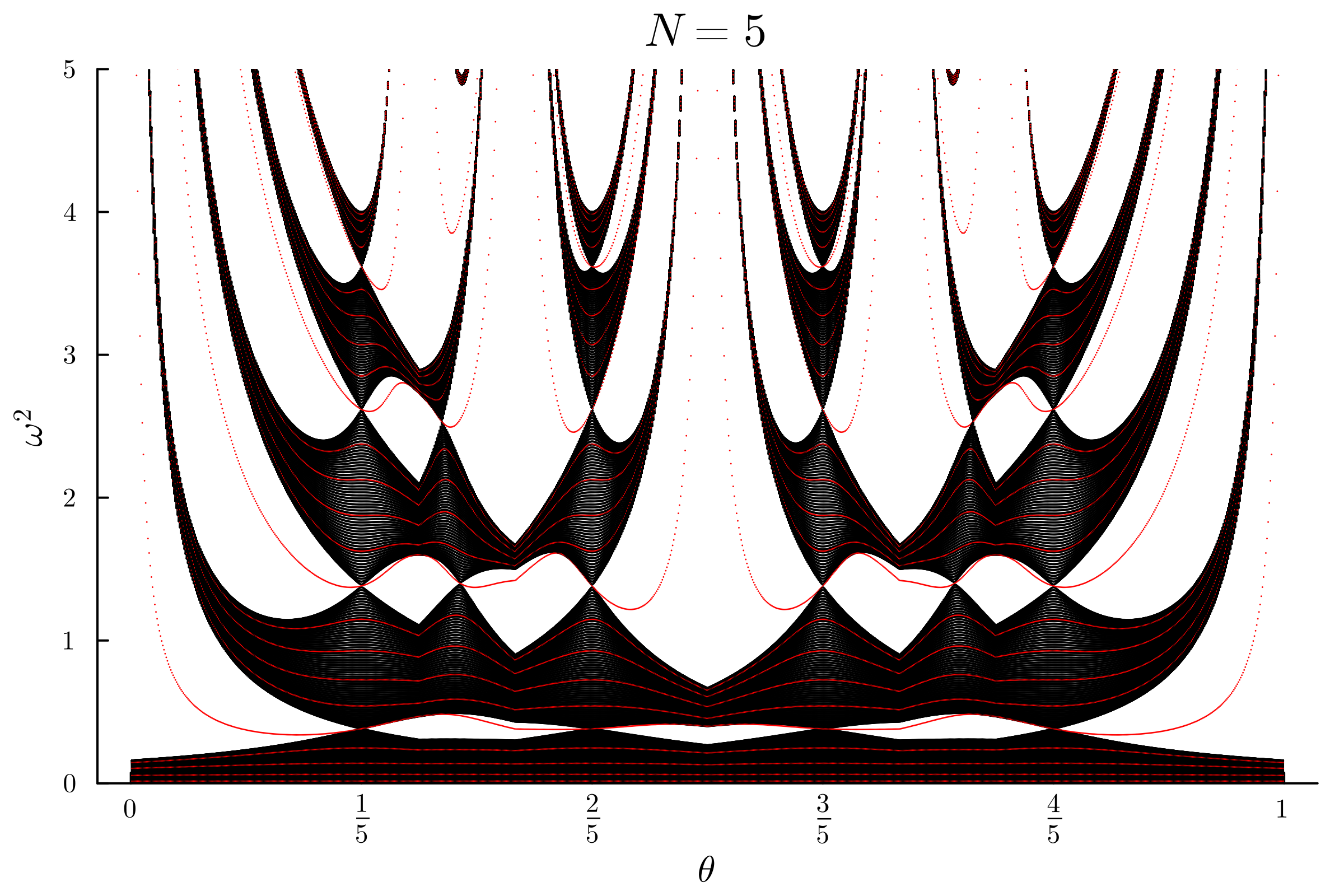}
    \caption{$N=5$}
    \label{fig:figure_3_panel_b}
\end{subfigure}
\begin{subfigure}{0.45\textwidth}
    \includegraphics[width=\textwidth,trim=0 0 0 3cm, clip]{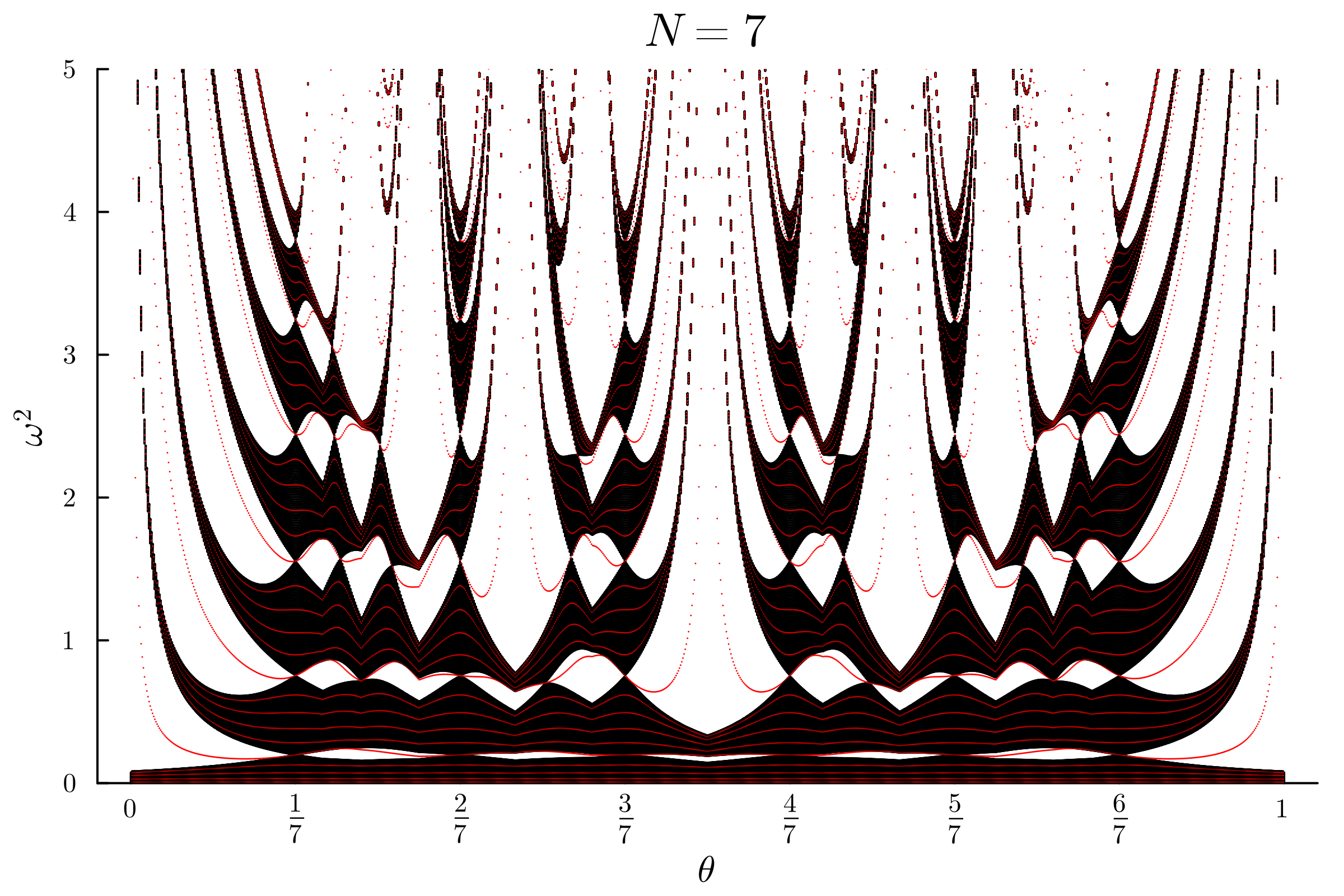}
    \caption{$N=7$}
    \label{fig:figure_3_panel_c}
\end{subfigure}
\begin{subfigure}{0.45\textwidth}
    \includegraphics[width=\textwidth,trim=0 0 0 3cm, clip]{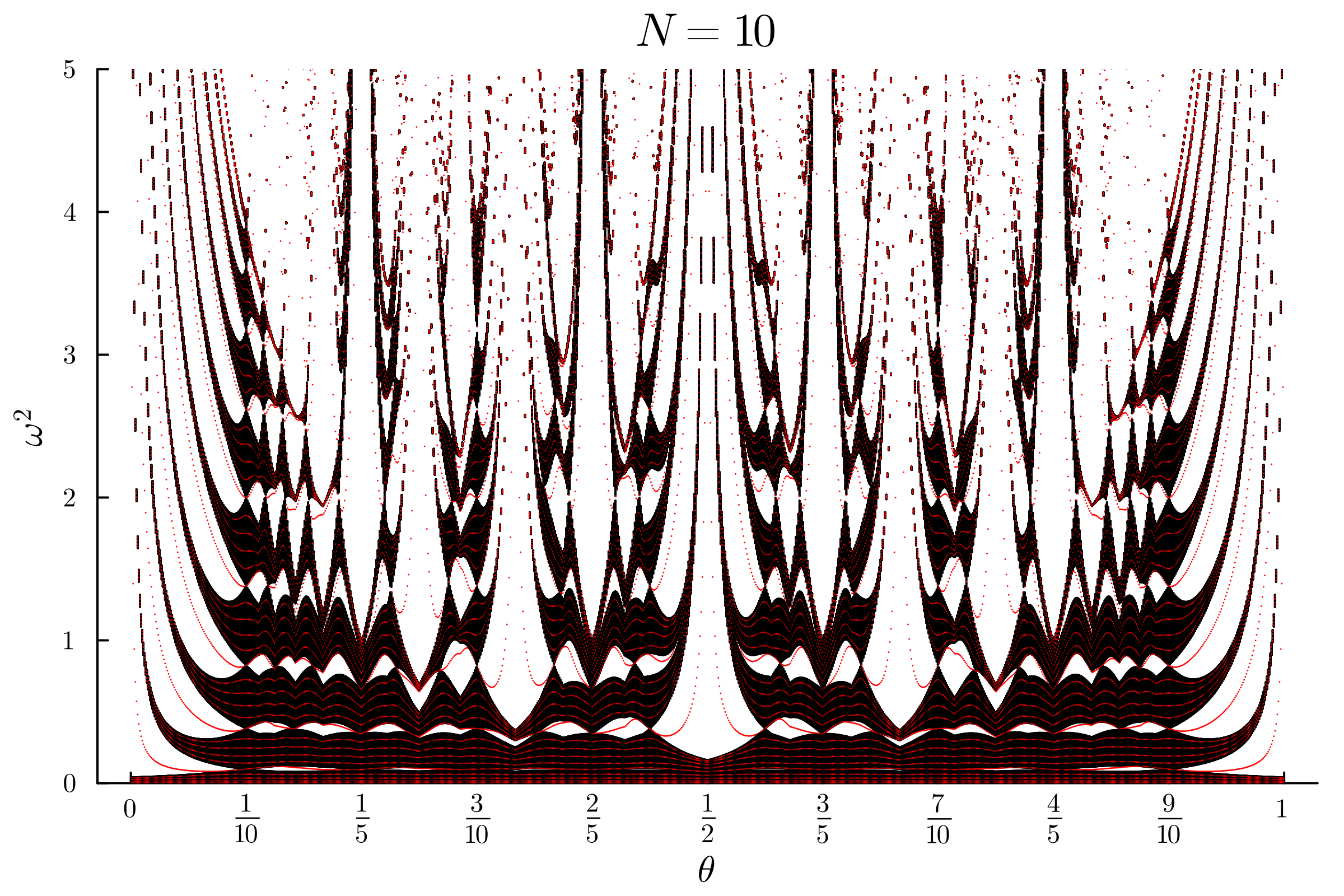}
    \caption{$N=10$}
    \label{fig:figure_3_panel_d}
\end{subfigure}
\caption{Phase space diagrams showing eigenfrequencies $\omega^2$ vs. modulation parameter $\theta$ for different $N$ values. Black lines represent the spectra of the infinite systems, while red lines indicate the eigenfrequencies of the finite problem with $n = 5$ unit cells. Complexity of the space diagram quickly grows with increasing number of masses in the unit cell.}
\label{fig:figure_3}
\end{figure}

The phase space diagram in Figure~\ref{fig_2}(a) is specific for a fixed number of points $N$ in the original unit cell. Different phase space diagrams can be computed for different $N$ (and these can be expected to become more intricate as $N$ increases). Figure~\ref{fig:figure_3} shows four different examples of phase space diagrams, for values ranging from $N = 3$ to $N = 10$. Panel (a) is $N = 3$; it is the simplest of the four diagrams, both in terms of edge states and conditions for overlapping of finite and infinite spectra. This system can be described by the established theory for SSH3 systems \cite{anastasiadis2022bulk, martinez2019edge, verma2024emergent, zhang2021topological}. Note that, in spite of the three-gap theorem, it is only when $N=3$ that the system coincides with the SSH3 model. This is shown in Appendix~\ref{app:SSH}. (The analogous property is true for the classic 2-periodic SSH system, which will only occur in the $N=2$ modulation - see Appendix~\ref{app:SSH} for details.) Complexity in the phase space diagram increases with the number of points $N$;  Figures~\ref{fig:figure_3}(b), \ref{fig:figure_3}(c) and \ref{fig:figure_3}(d) show the phase space diagrams for increasingly large $N$ and a correspondingly increasing number of edge states and stop bands are observed to continuously open and close as $\theta$ is modulated.   

% In the following, we will focus on determining those geometrical configurations for which the spectrum of the Dirichlet problem is a subset of the Floquet-Bloch spectrum.

\subsection{Vanishing of localised modes in the three-gap materials}
\label{sec:DeterminationOfFloquetBlochAngles}

To understand the existence of localised edge modes, we want to keep track of the configurations for which the localised modes become delocalised (in the sense that their localisation factor $\alpha(\omega)$ is equal to $1$).

Firstly, this trivially happens when the system reduces to a singly periodic system with uniform spring constants, as the band gaps all close for these particular values of $\theta$. These values of $\theta$ are given, for a fixed $N$, by \(\theta=\frac{i}{j}\) where \(j \in \left\{1, \dots, N\right\}\) and \(i \in \left\{1, \dots, j\right\}\). Since we have uniform neighbouring distances, the band gaps all close and \(\lvert \alpha(\omega) \rvert=1\) for every \(\omega^2 \in \sigma(K)\). Examples of these angles can be seen in Figure~\ref{fig:figure_3}: for example, when $N=3$ in Figure~\ref{fig:figure_3}(a) we see that the band gaps close when $\theta=1/3$ and $\theta=2/3$.

From examining Figures~\ref{fig_2} and~\ref{fig:figure_3}, it is apparent that equivalence to a singly periodic system is not the only mechanism through which localised modes become delocalised. In general, an interesting question is, for fixed $N\geq2$, to find those angles \(\theta\) such that there exists an \(\omega^2 \in \sigma(K(\theta))\) giving \(\lvert \alpha(\omega, \theta) \rvert=1\). In the following result, we characterise a broad class of such angles.

\begin{proposition} \label{prop:alpha1}
    For even \(N\), if \(\theta=\frac{N-1}{2N}\), then there are eigenvalues \(\omega^2_j=2-2\cos{\frac{2\pi j}{N}}\), \(j \in \left\{1, \dots, N/2-1\right\}\), of $K$ (the single unit cell system) which are such that \(\lvert \alpha(\omega_j) \rvert=1\).
\end{proposition}

\begin{proof}
    In this case, the matrix for single unit cell system in (\ref{single unit cell}) has the block form
    \[K=N
    \begin{pmatrix}
        A^{(3)}_{\frac{N}{2}-1} & -1\\
        -1 & \frac{5}{3} & -\frac{2}{3}\\
        & -\frac{2}{3} &  A^{(\frac{5}{3})}_{\frac{N}{2}-1}
    \end{pmatrix},
    \]
    where $A^{(a)}_m$ is the \(m \times m\) matrix
    \[A^{(a)}_m=
    \begin{pmatrix}
        a & -1 & 0 & \cdots & 0 & 0 \\
-1 & 2 & -1 & \cdots & 0 & 0 \\
0 & -1 & 2 & \cdots & 0 & 0 \\
\vdots & \vdots & \vdots & \ddots & \vdots & \vdots \\
0 & 0 & 0 & \cdots & 2 & -1 \\
0 & 0 & 0 & \cdots & -1 & 2
    \end{pmatrix}.\]
    We show that $\omega^2_j=2-2\cos{\frac{2\pi j}{N}}, \, j \in \{1, \dots, N/2-1\}$ are eigenvalues of \(K\) with the property that \(\lvert \alpha(\omega_j) \rvert=1\) by constructing the corresponding eigenvectors directly. It suffices to construct a vector \(\boldsymbol{v}_j=(v^j_0, v^j_1, \dots, v^j_{N-2})^T\) such that
    \[\left(\frac{1}{N}K-\omega^2_j I_{N-1}\right)\boldsymbol{v}_j=\boldsymbol{0}.\]
    That is, $\boldsymbol{v}_j$ is in the kernel of the tri-diagonal matrix $\frac{1}{N}K-\omega^2_j I_{N-1}$. 
%     \[\frac{1}{N}K-\omega^2_j I_{N-1}=
%     \begin{pmatrix}
%     1+2\cos{\frac{2\pi j}{N}} & -1\\
%        -1 & 2\cos{\frac{2\pi j}{N}} & -1\\
% & -1 & \ddots & \ddots\\
% & & \ddots & 2\cos{\frac{2\pi j}{N}} & -1\\
% & & & -1 & -\frac{1}{3}+2\cos{\frac{2\pi j}{N}} & -\frac{2}{3}\\
% & & & & -\frac{2}{3} & \ddots & \ddots\\
% & & & & & \ddots & 2\cos{\frac{2\pi j}{N}} & -1\\
% & & & & & & -1 & 2\cos{\frac{2\pi j}{N}}
%     \end{pmatrix}
%     \]
    Since \(d_0=\frac{1}{2N}\), \(d_{N-1}=\frac{1}{N}\), we need to show
    \begin{equation} \label{eq:goal}
        \left\lvert \frac{v^j_0}{v^j_{N-2}} \right\rvert=\frac{d_0}{d_{N-1}}=\frac{1}{2}
    \end{equation}
    so that $\lvert \alpha(\omega) \rvert=1$.
    
    Without loss of generality suppose that
    \begin{equation} \label{eq:v0}
        v^j_0=1 \qquad\text{and}\qquad v^j_1=1+2\cos{\frac{2\pi j}{N}}
    \end{equation}
    then
    \begin{equation}
    \label{cos_diff_eq}
        v^j_n=2\cos{\bigg(\frac{2\pi j}{N}\bigg)}v^j_{n-1}-v^j_{n-2}
    \end{equation}
    for \(n \in \{2, \dots, N-2\} \setminus \{\frac{N}{2}-1, \frac{N}{2}\}\).
    By induction, one can show that 
    \begin{equation}
        v^j_n=1+2\sum_{k=1}^n \cos{\bigg(k\frac{2\pi j}{N}\bigg)}
    \end{equation}
    for \(n \in \{2, \dots, \frac{N}{2}-1\}\). Then we have
    \begin{align*}
        v^j_{\frac{N}{2}}&=\cos{\pi j}+1+2\sum_{k=1}^{\frac{N}{2}} \cos{\bigg(k\frac{2\pi j}{N}\bigg)}\\
         v^j_{\frac{N}{2}+1}&=2\cos{\frac{2\pi j}{N}}\cos{\pi j}-\cos{\pi j}+1+2\sum_{k=1}^{\frac{N}{2}+1} \cos{\bigg(k\frac{2\pi j}{N}\bigg)}.
    \end{align*}
    By induction on (\ref{cos_diff_eq}) again, we have
    \begin{align*}
        v^j_{\frac{N}{2}+n} &= 1+2\sum_{k=1}^{\frac{N}{2}+n} \cos{\bigg(k\frac{2\pi j}{N}\bigg)}+\sum_{k=-n}^n (-1)^{k+n}\cos{\bigg(\bigg(\frac{N}{2}+k\bigg)\frac{2\pi j}{N}\bigg)}\\
        % &= 1+2\sum_{k=1}^{\frac{N}{2}-n-1} \cos{(k\frac{2\pi j}{N})}+2\sum_{k=\frac{N}{2}-n}^{\frac{N}{2}+n} \cos{(k\frac{2\pi j}{N})}+\sum_{k=-n}^n (-1)^{k+n}\cos{((\frac{N}{2}+k)\frac{2\pi j}{N})}\\
        &= 1+2\sum_{k=1}^{\frac{N}{2}-n-1} \cos{\bigg(k\frac{2\pi j}{N}\bigg)}+\sum_{k=-n}^n (2+(-1)^{k+n})\cos{\bigg(\bigg(\frac{N}{2}+k\bigg)\frac{2\pi j}{N}\bigg)}\\
        &= 1+2\sum_{k=1}^{\frac{N}{2}-n-1} \cos{\bigg(k\frac{2\pi j}{N}\bigg)}+(-1)^j\sum_{k=-n}^n (2+(-1)^{k+n})\cos{\bigg(k\frac{2\pi j}{N}\bigg)}\\
         &= 1+2\sum_{k=1}^{\frac{N}{2}-n-1} \cos{\bigg(k\frac{2\pi j}{N}\bigg)}+(-1)^j\left(2+(-1)^n+2\sum_{k=1}^n (2+(-1)^{k+n})\cos{\bigg(k\frac{2\pi j}{N}\bigg)}\right),
    \end{align*}
    for \(n \in \{2, \dots, \frac{N}{2}-2\}\). In particular,
    \begin{align}
        v^j_{N-2} &= 1+2\cos{\frac{2\pi j}{N}}+(-1)^j\left(2+(-1)^{\frac{N}{2}}+2\sum_{k=1}^{\frac{N}{2}-2} (2+(-1)^{k+\frac{N}{2}})\cos{\bigg(k\frac{2\pi j}{N}\bigg)}\right) \nonumber \\
        &= 1+2\cos{\frac{2\pi j}{N}}+(-1)^j\Bigg[2+(-1)^{\frac{N}{2}}+2\left(\frac{\sin{((\frac{N}{2}-\frac{3}{2})\frac{2\pi j}{N})}}{\sin{\frac{\pi j}{N}}}-1\right) \nonumber \\&\hspace{6cm}+(-1)^{\frac{N}{2}}\left((-1)^{\frac{N}{2}}\frac{\cos{((\frac{N}{2}-\frac{3}{2})\frac{2\pi j}{N})}}{\cos{\frac{\pi j}{N}}}-1\right)\Bigg] \nonumber \\
        &= 1+2\cos{\frac{2\pi j}{N}}-\left(1+4\cos^2{\frac{\pi j}{N}}\right) \nonumber \\
        &= -2. \label{eq:vN-2}
    \end{align}
    Hence, comparing \eqref{eq:v0} and \eqref{eq:vN-2}, we have that
    \begin{equation}
        \alpha(\omega_j)=-\frac{d_0}{d_{N-1}}\frac{v^j_{N-2}}{v^j_0}= -\frac{(2N)^{-1}}{N^{-1}} \frac{-2}{1}=1.
    \end{equation}
    Finally, we need to check that
    \[2\cos{\bigg(\frac{2\pi j}{N}\bigg)}v^j_{N-2}-v^j_{N-3}=0.\]
    Indeed, it holds that
    \begin{align*}
        2\cos{\bigg(\frac{2\pi j}{N}\bigg)}v^j_{N-2}-v^j_{N-3} &= 1+2\sum_{k=1}^{N-1} \cos{\bigg(k\frac{2\pi j}{N}\bigg)}+\sum_{k=1}^{N-1} (-1)^{k-1}\cos{\bigg(k\frac{2\pi j}{N}\bigg)}\\
        &= \frac{\sin{((N-\frac{1}{2})\frac{2\pi j}{N})}}{\sin{\frac{\pi j}{N}}}-\frac{1}{2}\left((-1)^{N-1}\frac{\cos{((N-\frac{1}{2})\frac{2\pi j}{N})}}{\cos{\frac{\pi j}{N}}}-1\right)\\
        &= 0,
    \end{align*}
    which completes the proof.
\end{proof}

\begin{remark}
    The statement and proof of Proposition~\ref{prop:alpha1} for $\theta = (N-1)/2N$ when $N$ is an even number is equivalent for $\theta = (N-2)/2(N-1)$ for $N$ being an odd number.
\end{remark}

% \todo[inline]{Starting form $N = 7$, I think we also have $\theta = (N-2)/2N$ for $N$ being odd and $\theta = (N-3)/2(N-1)$ for $N$ being even. I don't know if the proof for this is difficult, but it might be interesting to give it a try in order to see how many eigenvalues for this situation have $|\alpha|=1$. If this is true, maybe we can conjecture that something special happens at $\theta = (N-i)/2N$.

% Sequence for $\theta = (N-2)/2N$ with $N=7$: A-A-B-A-C-A-B.

% Sequence for $\theta= (N-2)/2N$ with $N=9$: A-B-C-C-C-C-C-C-B-A.

% Sequence for $\theta = (N-2)/2N$ with $N=11$: A-A-B-A-B-A-C-A-B-A-B.

% Sequence for $\theta = (N-2)/2N$ with $N=13$: C-A-B-A-B-A-B-B-A-B-A-B-A.

% Sequence for $\theta = (N-2)/2N$ with $N=15$: A-A-B-A-B-A-B-A-C-A-B-A-B-A-B.

% At $\theta = (N-2)/2N$ the sequence of springs is the following: A-A- ($1/4(N-3)$B-A )-C-($1/4(N-3)$A-B ). This sequence is valid for $N = 7,11,15,19,...$. This structure is "equivalent" to the one at $\theta = (N-1)/2N$, which is: A-($1/2(N-2)$B)-C-($1/2(N-2)$B), valid for $N=4,6,8,10,...$, and it is sketched as A-B-C-B.

% In both cases, $B=2A$ and $C=3A$.}

\begin{figure}
\centering\includegraphics[width=0.8\linewidth]{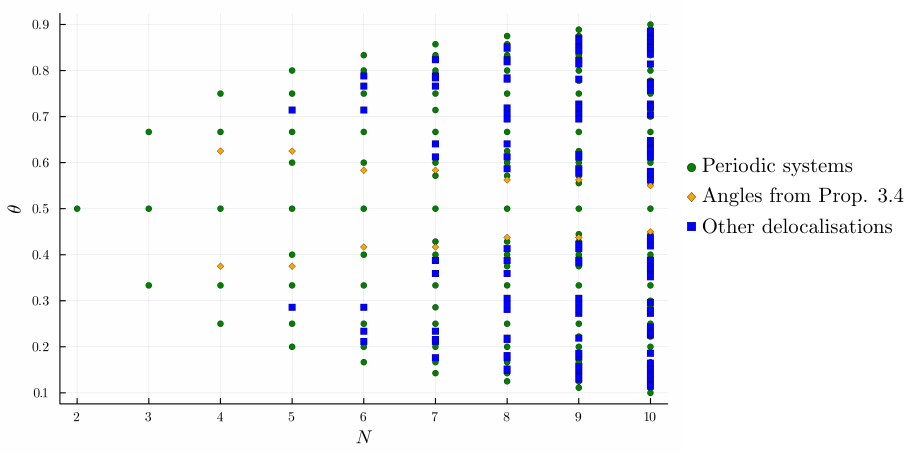}
\caption{Modulation angles at which at least one eigenmode of the single unit cell system $K$ has a localisation factor $|\alpha| = 1$ so is delocalised. Green circles show those configurations for which the resulting unit cell is made of all equal springs so is singly periodic. Golden diamonds represent configurations characterised in Proposition~\ref{prop:alpha1}, with $\theta = (N-1)/2N$ for even $N$ or $\theta = (N-2)/2(N-1)$ for odd $N$ (in which case half of the eigenvalues are known to be delocalised). Blue squares show other geometrical configurations for which $|\alpha| = 1$ occurs.}
\label{fig_4}
\end{figure}

Despite of the insight gained from Proposition~\ref{prop:alpha1}, there are still many angles for which one or more of the eigenvectors is delocalised (in the sense that it has $|\alpha|=1$). Figure~\ref{fig_4} shows, for each number of masses in the unit cell $N$, the values of $\theta$ for which at least one eigenvalue has $|\alpha| = 1$. It is noticeable that the number of modulation angles $\theta$ which lead to a delocalisation increases with $N$, such that we have a rich structure of degenerate angles for large $N$. 

Three different markers have been used to represent different subsets of the geometrical configurations of interest in Figure~\ref{fig_4}. First, green circles represent those geometrical configuration in which the resulting unit cell is made of all-equal springs so the system is singly periodic (all the eigenvalues belonging to $\sigma(K)$ have $|\alpha| = 1$, which means $\sigma(K) \subseteq \sigma(\tilde{K})$). Golden diamonds represent angles characterised in Proposition~\ref{prop:alpha1}, which are such that $\theta = (N-1)/2N$ when $N$ is an even number and $\theta = (N-2)/2(N-1)$ when $N$ is an odd number. As we have shown, in this case half of the eigenvalues have $|\alpha| = 1$. Finally, blue squares represent all the other situations in which at least one of the eigenvalues has $|\alpha|=1$. It is noticeable that for $N\leq4$ the mechanisms of periodicity or Proposition~\ref{prop:alpha1} are sufficient to explain all the angles of interest. However, other angles appear for larger $N$ and as $N$ increases the number of angles which lead to a delocalisation becomes large. These angles (the blue squares) haven't been characterised analytically in this work, and represent an interesting open question for future exploration. We speculate that blue squares might be related to rational numbers with denominators being multiples of $N$. However, numerical results for big $N$, such as $N =10$ in Figure~\ref{fig_4}, show blue squares with $\theta$ values associated being close to irrational numbers.

\section{Conclusion}

In this article, we have studied the existence of localised edge states in finite-sized mass-spring chains of resonators with patterns generated by an algorithm based on the three-gap theorem. This system is exciting as it allows us to create intricate heterogeneous structures with complex spectral properties using only three different components. This would be convenient for any subsequent experimental implementations. On top of this, this system serves as a potential avenue to bring together some of the most powerful ideas in wave physics, such as the topological theory of SSH systems and the spectral flow characterisation of Harper modulation. 

% Therefore, the topology of our system might be reduced to the one in SSH or in SSH-3, as shown in the Appendix \ref{app:SSH}. Furthermore, changes in the geometry of our system can be controlled using one continuous variable $\theta$, which allows us to represent spectra in a phase space diagram ($\omega^2$ vs $\theta$) and study the behaviour of the localised states as adiabatic changes of the modulation parameter, which, in turns, connects our system with the notion of spectral flow. 

The main theoretical tool used in our work is the localisation factor, that measures the growth or decay of an eigenvector across a unit cell and offers a simple but powerful metric to study the localisation properties of the edge states. This versatile quantity captures the existence of localised edge modes, which occur if and only if $|\alpha(\omega)|\neq 1$. In doing so, it captures the transition between localised and extended states in a one-dimensional system \emph{cf.} \cite{ammari2024exponentially}. In particular, this unifying tool characterises the edge modes in SSH and SSH3 systems and also captures the spectral flow of edge modes as the modulation parameter is varied (as we must have $|\alpha(\omega)|=1$ when edge modes enter or leave a band gap).

% ultimately unifies topological indices and notions of spectral flow in one-dimensional systems. As the main results for the \textit{localisation factor} applied to any finite system (where we have studied Dirichlet boundary conditions at both ends of the finite chain), we have Theorems \ref{thm:alpha1} and \ref{Theorem:alphaGap}, stating that an eigenstate of a finite system will be localised in space if and only if $|\alpha(\omega)|\neq 1$, and that all those eigenstates with $|\alpha(\omega)|\neq 1$ lie in the gap of the periodic structure.

This study suggests several interesting directions for future work. We proved results that describe a subset of the pairs $(\theta,N)$ of parameter values at which eigenmodes are delocalised (through either simple periodicity emerging or Proposition~\ref{prop:alpha1}). However, our numerical results show that the set of these parameter values is much larger than those characterised in this work, especially when $N$ is large. Understanding these parameter values likely hinges on understanding the symmetries that can exist in the system, but we leave this as an open question for future investigation. Another interesting question to explore is whether other well-known systems appear from our three-gap algorithm for certain parameter values (as well as the SSH and SSH3 systems that emerge when $N=2$ and $N=3$, respectively). For example, are there values of $(\theta,N)$ such that Fibonacci tilings appear? These are a particularly widely studied example of systems generated by tiling rules, that are composed of a just two different building blocks and have spectral with exotic fractal properties \cite{jagannathan2021fibonacci, davies2024super, dal2007spectral}.

\acknowledgements{Marc Martí-Sabaté acknowledges financial support through the DYNAMO project (101046489), funded by the European Union, but the views and opinions expressed are, however, those of the authors only and do not necessarily reflect those of the European Union or the European Innovation Council. Neither the European Union nor the granting authority can be held responsible for them.}

\appendix
\section{Relation to SSH and SSH-3 models} \label{app:SSH}

A natural question to ask about the systems considered here, generated by the three-gap algorithm, is whether certain values of $N$ and $\theta$ give rise to other, well-known systems. For example, do some of the systems in fact coincide with SSH-2 or SSH-3 models? In the following we will show that SSH-2 and SSH-3 can only appear when $N=2$ and $N=3$, respectively. For the case of SSH-3 this result is particularly noteworthy, as the number of different coupling strengths (spring constants) is restricted to 3 by the three-gap theorem, but it turns out that the pattern never coincides with SSH-3 (as in \cite{anastasiadis2022bulk, martinez2019edge, verma2024emergent, zhang2021topological}) for $N>3$.

\begin{definition}
Fix \(n \in \mathbb{N}_{\geq 2}\) and let \(\{y_i\}_{i=0}^{kn-1} \subseteq [0,1)\) be strictly increasing for some \(k \in \mathbb{N}\) such that \(y_0=0\). Define \(d_i=y_{i+1}-y_i\) for \(i \in \{0,\dots, kn-1\}\) where \(y_{kn}=1\). We say \(\{y_i\}_{i=0}^{kn-1}\) is an \textit{SSH-\(n\) model} if
\begin{enumerate}
    \item \(d_{i+jn}= d_i\) for all \(i \in \{0,\dots, n-1\}\) and \(j \in \{1, \dots, k-1\}\)
    \item \(\{d_i\}_{i=0}^{n-1}\) is not a singleton
\end{enumerate}
\end{definition}

Fix \(N \in \mathbb{N}_{\geq 2}\), \(\theta \in (0,1)\). Let \(\{x_i\}_{i=0}^{N-1}\) be an increasing list representing the set \(\{\mathrm{frac}(j\theta):j=0, \dots, N-1\}\). Define \(d_i=x_{i+1}-x_i\) for \(i \in \{0,\dots, N-1\}\) where \(x_N=1\). Note that for irrational \(\theta\), \(\{x_i\}_{i=0}^{N-1}\) is always distinct. For rational \(\theta=\frac{p}{q}\) where \(p, q \in \mathbb{N}\) are coprime, if \(N=q\), then \(x_i=\frac{i}{q}\) for \(i \in \{0, \dots, q-1\}\). Note that in the definition of the SSH model, we want \(\{x_i\}_{i=0}^{N-1}\) to be distinct (strictly increasing) for convenience. Thus, for this rational case, we want \(N \leq q\).

\begin{proposition}
    The list \(\{x_i\}_{i=0}^{N-1}\) is an SSH-\(2\) model if and only if \(N=2, \theta \neq \frac{1}{2}.\)
\end{proposition}
\begin{proof}
    If \(N=2, \theta \neq \frac{1}{2}\), we can see that \(\{x_i\}_{i=0}^{1}\) is an SSH-\(2\) model since there are only 2 points and \(d_0 \neq d_1\). Now suppose \(\{x_i\}_{i=0}^{N-1}\) is an SSH-\(2\) model. If \(N=2\), then \(\theta \neq \frac{1}{2}\) so that \(d_0 \neq d_1\). Then it suffices to show that for \(N>2\), \(\{x_i\}_{i=0}^{N-1}\) is not an SSH-\(2\) model for any \(\theta \in (0,1)\).
    
    For the sake of contradiction suppose there exists \(\theta \in (0,1)\) such that \(\{x_i\}_{i=0}^{N-1}\) is SSH-\(2\) where \(N>2\). Then \(N\) is even and \(\frac{N}{2}(d_0+d_1)=1\). The list \[\{x_i\}_{i=0}^{N-1}=\left\{0,d_0,d_0+d_1,d_0+d_1+d_0, \dots, (\frac{N}{2}-1)(d_0+d_1)+d_0\right\}\]
    
    Case 1: \(\theta=k(d_0+d_1)\) for some \(k \in \{1, \dots, \frac{N}{2}-1\}\). 
    Then \(j\theta \ (\mathrm{mod} 1)=j\cdot k(d_0+d_1) \ (\mathrm{mod} \frac{N}{2}(d_0+d_1))\) is a multiple of \(d_0+d_1\) for all \(j\). Hence \(m(d_0+d_1)+d_0 \notin \{x_i\}_{i=0}^{N-1}\) for any \(m \in \{0, \dots, \frac{N}{2}-1\}\), contradiction.
    
    Case 2: \(\theta=l(d_0+d_1)+d_0\) for some \(l \in \{0, \dots, \frac{N}{2}-1\}\). 
    Without loss of generality suppose \(d_0>d_1\). Then
    \begin{align*}
        2\theta \ (\mathrm{mod} 1) &= (2l+1)(d_0+d_1)+d_0-d_1 \ (\mathrm{mod} 1)\\
        &=(2l+1 \ (\mathrm{mod} \frac{N}{2}))(d_0+d_1)+d_0-d_1
    \end{align*}
    Hence \(d_0-d_1=d_0\), which implies \(d_1=0\), contradiction.
\end{proof}

\begin{proposition}
    The list \(\{x_i\}_{i=0}^{N-1}\) is an SSH-\(3\) model if and only if \(N=3, \theta \neq \frac{1}{3},\frac{2}{3}.\)
\end{proposition}
\begin{proof}
    It suffices to show that for \(N>3\), \(\{x_i\}_{i=0}^{N-1}\) is not an SSH-\(3\) model for any \(\theta \in (0,1)\).  For the sake of contradiction suppose that there exists \(\theta \in (0,1)\) such that \(\{x_i\}_{i=0}^{N-1}\) is SSH-\(3\) where \(N>3\). Then \(N\) is a multiple of \(3\) and \(\frac{N}{3}L=1\), where \(L=d_0+d_1+d_2\). The list \[\{x_i\}_{i=0}^{N-1}=\left\{0,d_0,d_0+d_1, L, L+d_0, L+d_0+d_1, \dots, (\frac{N}{3}-1)L+d_0+d_1\right\}\]
    
    We consider three different cases for the value of $\theta$. 
    Case 1 is that \(\theta=kL\) for some \(k \in \{1, \dots, \frac{N}{3}-1\}\). Then \(j\theta \ (\mathrm{mod} 1)=j\cdot kL \ (\mathrm{mod} \frac{N}{3}L)\) is a multiple of \(L\) for all \(j\). Hence \(mL+d_0, mL+d_0+d_1 \notin \{x_i\}_{i=0}^{N-1}\) for any \(m \in \{0, \dots, \frac{N}{3}-1\}\), contradiction.
    
    Case 2 is that \(\theta=kL+d_0\) for some \(k \in \{0, \dots, \frac{N}{3}-1\}\). We have
    \[2\theta \ (\mathrm{mod} 1) = (2k+1)L+d_0-d_1-d_2 \ (\mathrm{mod} 1)\]
    Since \(|d_0-d_1-d_2|<L\), \[d_0-d_1-d_2 \in \left\{0, d_0+d_1-L\right\}.\]
    If \(d_0-d_1-d_2=0\), then \(d_0=\frac{L}{2}\)
    \begin{align*}
        j\theta \ (\mathrm{mod} 1) &= j(kL+d_0) \ (\mathrm{mod} 1)\\
        &= (jk+\frac{j}{2})L \ (\mathrm{mod} \frac{N}{3}L)
    \end{align*}
    Thus for even \(j\),
    \[j\theta \ (\mathrm{mod} 1)=(jk+\frac{j}{2} \ (\mathrm{mod} \frac{N}{3})) \cdot L\]
    and for odd \(j\),
    \[j\theta \ (\mathrm{mod} 1)=(jk+\frac{j-1}{2} \ (\mathrm{mod} \frac{N}{3})) \cdot L + d_0.\]
    Hence, \(mL+d_0+d_1 \notin \{x_i\}_{i=0}^{N-1}\) for any \(m \in \{0, \dots, \frac{N}{3}-1\}\), contradiction. Conversely, if \(d_0-d_1-d_2=d_0+d_1-L\), then \(d_0=d_1\),
    \begin{align*}
        3\theta \ (\mathrm{mod} 1) &= 2\theta + \theta \ (\mathrm{mod} 1)\\
        &= (2k+1)L-d_2+kL+d_0 \ (\mathrm{mod} 1)\\
        &= (3k+1)L+d_0-d_2 \ (\mathrm{mod} 1).
    \end{align*}
    Hence \(d_0-d_2=0\), we have \(d_0=d_1=d_2\), contradiction.
    
    Case 3 is \(\theta=kL+d_0+d_1\) for some \(k \in \{0, \dots, \frac{N}{3}-1\}\). In this case, we similarly have
    \[2\theta \ (\mathrm{mod} 1) = (2k+1)L+d_0+d_1-d_2 \ (\mathrm{mod} 1).\]
    Since \(|d_0+d_1-d_2|<L\), \[d_0+d_1-d_2 \in \left\{0, d_0\right\}.\]
    If \(d_0+d_1-d_2=0\), then \(d_2=\frac{L}{2}\). Similarly to Case 2, we have \(mL+d_0 \notin \{x_i\}_{i=0}^{N-1}\) for any \(m \in \{0, \dots, \frac{N}{3}-1\}\), contradiction. 
    If \(d_0+d_1-d_2=d_0\), then \(d_1=d_2\),
    \begin{align*}
        3\theta \ (\mathrm{mod} 1) &= (2k+1)L+d_0+kL+d_0+d_1 \ (\mathrm{mod} 1)\\
        &= (3k+2)L+d_0-d_2 \ (\mathrm{mod} 1).
    \end{align*}
    Hence \(d_0-d_2=0\), we have \(d_0=d_1=d_2\), contradiction.
\end{proof}

\bibliography{references} %%%%% .Bib file
\bibliographystyle{ieeetr}

\end{document}

%% file: authors.tex
\author{Yinglai Wang}
\affiliation{
    Department of Mathematics, Imperial College London, London SW7 2AZ, UK
}

\author{Bryn Davies}
\affiliation{
    Mathematics Institute, University of Warwick, Coventry CV4 7AL, UK
}

\author{Marc Mart{\' i} Sabat{\' e}}
\email{m.marti-sabate23@imperial.ac.uk}
\affiliation{
    Department of Mathematics, Imperial College London, London SW7 2AZ, UK
}